\documentclass[a4paper,11pt]{article}
\usepackage{amsmath}
\usepackage{amsthm}
\usepackage{amsfonts}
\usepackage{mathrsfs}
\usepackage{amssymb}
\usepackage{mathpazo}
\usepackage[dvipsnames]{xcolor}
\usepackage{bm}
\usepackage{paralist}
\usepackage{natbib}
\usepackage[title]{appendix}%
\usepackage{url}
\usepackage[textwidth=6.1in,textheight = 9in]{geometry}
\usepackage{placeins}
\usepackage[hidelinks]{hyperref}

\usepackage{multirow, float}
\usepackage{makecell}
\usepackage{calc}
\usepackage{graphicx}
\usepackage{enumitem}
\usepackage{microtype}
\usepackage{todonotes}
\usepackage{caption}
\DeclareCaptionStyle{italic}[justification=centering]
{labelfont={bf},textfont={it},labelsep=colon}
\usepackage[compact]{titlesec}
\titlespacing{\section}{0pt}{2ex}{1ex}
\titlespacing{\subsection}{0pt}{1ex}{0ex}
\titlespacing{\subsubsection}{0pt}{0.5ex}{0ex}
\usepackage{titling}\pretitle{\begin{center}\LARGE\bfseries}
\usepackage[bottom,hang,flushmargin]{footmisc}
\usepackage{adjustbox}

\allowdisplaybreaks
\usepackage{caption}
\usepackage{subcaption}

\usepackage{longtable}
\usepackage{booktabs}
\usepackage{array}
\newcolumntype{M}[1]{>{\centering\arraybackslash}m{#1}}
\newcolumntype{L}[1]{>{\raggedright\arraybackslash}m{#1}}
\newcolumntype{R}[1]{>{\raggedleft\arraybackslash}m{#1}}

\newcommand{\etavet}{\bm{\eta}}

\newcommand{\Unovet}{\bm{1}}

\newcommand{\bvet}{\bm{b}}

\newcommand{\xvet}{\bm{x}}

\newcommand{\Bvet}{\bm{B}}
\newcommand{\Cvet}{\bm{C}}

\newcommand{\Evet}{\bm{E}}

\newcommand{\Gvet}{\bm{G}}
\newcommand{\Hvet}{\bm{H}}
\newcommand{\Ivet}{\bm{I}}

\newcommand{\Kvet}{\bm{K}}

\newcommand{\Mvet}{\bm{M}}

\newcommand{\Pvet}{\bm{P}}
\newcommand{\Qvet}{\bm{Q}}

\newcommand{\Svet}{\bm{S}}

\newcommand{\Uvet}{\bm{U}}
\newcommand{\Wvet}{\bm{W}}
\newcommand{\Xvet}{\bm{X}}

\newcommand{\Zerovet}{\bm{0}}
\newcommand{\Omegavet}{\bm{\Omega}}
\newcommand{\Sigmavet}{\bm{\Sigma}}

\theoremstyle{definition}
\newtheorem{theorem}{Theorem}
\usepackage[affil-it]{authblk}
\newcommand{\email}[1]{\affil{Email: {\upshape\href{mailto:#1}{\texttt{#1}}}}}

\makeatletter
\renewenvironment{abstract}{%
	\if@twocolumn
	\section*{\abstractname}%
	\else %
	\begin{center}%
		{\bfseries \large\abstractname\vspace{\z@}}%
	\end{center}%
	\quotation
	\fi}
{\if@twocolumn\else\endquotation\fi}
\makeatother

\usepackage{tikz}

\usetikzlibrary{arrows,shapes,positioning,shadows,trees,calc,fit}
\usetikzlibrary{matrix, decorations.pathreplacing}
\tikzset{
  basic/.style  = {draw, text width=2cm, drop shadow, font=\sffamily, rectangle},
  root/.style   = {basic, rounded corners=2pt, thin, align=center,
                   fill=green!30},
  level 2/.style = {basic, rounded corners=6pt, thin,align=center, fill=green!60,
                   text width=4em},
  level 3/.style = {basic, thin, align=left, fill=pink!60, text width=1.5em}
}
\newcommand{\relation}[3]
{
	\draw (#3.south) -- +(0,-#1) -| ($ (#2.north) $)
}

\makeatletter
\newcommand{\githuburl}{\url{https://github.com/danigiro/ct-comp}.}
\makeatother

\title{Insights into regression-based cross-temporal forecast reconciliation}

\author{Daniele Girolimetto}
\author{Tommaso Di Fonzo}
\affil{Department of Statistical Sciences, University of Padua, Padova 35121, Italy}
\email{daniele.girolimetto@unipd.it}

\date{\today}
\begin{document}
	
\def\spacingset#1{\renewcommand{\baselinestretch}{#1}\small\normalsize}
\spacingset{1.1}

\thispagestyle{empty} \clearpage\maketitle

\begin{abstract}
	\noindent Cross-temporal forecast reconciliation aims to ensure consistency across forecasts made at different temporal and cross-sectional levels. We explore the relationships between sequential, iterative, and optimal combination approaches, and discuss the conditions under which a sequential reconciliation approach (either first-cross-sectional-then-temporal, or first-temporal-then-cross-sectional) is equivalent to a fully (i.e., cross-temporally) coherent iterative heuristic. Furthermore, we show that for specific patterns of the error covariance matrix in the regression model on which the optimal combination approach grounds, iterative reconciliation naturally converges to the optimal combination solution, regardless the order of application of the uni-dimensional cross-sectional and temporal reconciliation approaches. Theoretical and empirical properties of the proposed approaches are investigated through a forecasting experiment using a dataset of hourly photovoltaic power generation. The study presents a comprehensive framework for understanding and enhancing cross-temporal forecast reconciliation, considering both forecast accuracy and the often overlooked computational aspects, showing that significant improvement can be achieved in terms of memory space and computation time, two particularly important aspects in the high-dimensional contexts that usually arise in cross-temporal forecast reconciliation.
\end{abstract}

\begin{itemize}[nosep, align=left, leftmargin = !]
	\item[\textbf{Keywords}] \textit{Coherent forecasts; Cross-temporal forecast reconciliation; Sequential, iterative, and optimal combination approaches}
\end{itemize}
\vfill

\newpage
\spacingset{1.3}

\section{Introduction}
\label{sec:Introduction}
A large number of observed phenomena require accurate forecasts across different dimensions, e.g., different time periods and granularities, geographical regions, customer groups. Achieving coherent forecasts at any level of (dis)aggregation is challenging but crucial for optimal decision making, resource allocation and operational efficiency. As effective decision making depends on the support of accurate and coherent forecasts, forecast reconciliation, originally intended as a specific tool to improve the accuracy of the forecasts of a hierarchical time series \citep{Fliedner2001, Hyndman2011, Wickramasuriya2019}, has then become a key post-forecasting process aimed at improving the accuracy of forecasts for general linearly constrained multiple time series, of which hierarchical and grouped time series are a specific special case \citep{Panagiotelis2021, Girolimetto2024-ft}. The need for reconciliation arises when base forecasts of the components of a multivariate variable, often produced independently by separate organizational units or using different models, do not match the internal constraints applicable to that variable \citep{Kourentzes2022}.

Classical reconciliation methods for hierarchical time series, such as bottom-up and top-down, have been widely employed to address cross-sectionally incoherent forecasts. However, as \cite{Pennings2017} point out, these classical methods frequently overlook valuable information from different levels of the hierarchy. The introduction of regression-based reconciliation techniques in recent years has addressed this limitation. Consequently, hierarchical forecasting has significantly evolved to include modern least squares-based reconciliation techniques in the cross-sectional framework \citep{Hyndman2011, Wickramasuriya2019, Panagiotelis2021}, which were later extended for temporal hierarchies \citep{Athanasopoulos2017, Yang2017te, Spiliotis2019, Nystrup2020, Bergsteinsson2021, Nystrup2021}. The temporal coherence aspect is a natural extension, requiring the reconciliation of forecasts with different time granularities, such as months, quarters, or years, to ensure consistency \citep{Athanasopoulos2017}.

The cross-temporal framework combines cross-sectional and temporal dimensions to achieve fully coherent forecasts \citep{Kourentzes2019, Yagli2019, Punia2020, Spiliotis2020AE}. In this setting, complexity and dimensions of the problem grow very quickly along with the requested computational time and memory. \cite{Kourentzes2019} and \cite{Yagli2019} both tackle this issue using sequential procedures that alternate uni-dimensional reconciliation approaches. On the other hand, \cite{DiFonzoGiro2023} proposed an iterative heuristic, extending the results so far, and an optimal (in the least squares sense) cross-temporal forecast reconciliation approach, that simultaneously encompasses both  contemporaneous and temporal coherency. 

In recent years, there has been a significant increase in the application of cross-temporal reconciliation techniques to improve the accuracy of forecasting models by addressing inconsistencies across different time horizons. Notable applications of these methods include the works of \cite{Sanguri2022}, \cite{DiFonzoGiro2023SE}, \cite{English2024}, \cite{Quinn2024}. Moreover, \cite{Girolimetto2024-jm} proposed a probabilistic cross-temporal linear reconciliation technique, offering a more flexible approach to consider forecast uncertainties. Finally, it is worth mentioning the contribution by \cite{Rombouts2024}, who recently proposed the first Machine Learning-based cross-temporal forecast reconciliation technique, and performed an empirical comparison with the standard linear approaches.

In this paper, we address some open-issues related to the relationships between sequential, iterative and optimal combination cross-temporal forecast reconciliation approaches.  We discuss the conditions under which a sequential (either first-cross-sectional-then-temporal, or first-temporal-then-cross-sectional) approach, is equivalent to a fully (i.e, cross-temporally) coherent iterative heuristic. We also show that, for specific patterns of the error covariance matrix of the regression model on which the optimal combination approach grounds, an iterative reconciliation procedure `converges' to the optimal combination reconciliation solution regardless the order with which each uni-dimensional reconciliation is applied. The reduction of the computing effort is evaluated in the experiment of forecasting hourly photovoltaic power generation considered by \cite{Yagli2019} and  \cite{DiFonzoGiro2022, DiFonzoGiro2023SE}.

The paper is organised as follows. In \autoref{sec:recap}, we briefly recapitulate the main definitions and results for cross-sectional, temporal and cross-temporal forecast reconciliation, following the notation presented in \cite{Athanasopoulos2024} and extended in \cite{Girolimetto2024-jm}. In \autoref{sec:theorems}, we present two new results, in the form of theorems describing the relationships between sequential, iterative and optimal cross-temporal reconciliation approaches. To assess the impact of the new results on memory space and computation time, a forecasting experiment on an hourly power generation dataset is presented in \autoref{sec:forexp}. Conclusions follow in \autoref{sec: conclusion}. All the forecast reconciliation procedures considered in this paper are available in the R package \texttt{FoReco} \citep{FoReco2024}. 

\section{Regression-based cross-temporal reconciliation}
\label{sec:recap}

In this section, we briefly describe the optimal combination cross-temporal point forecast reconciliation methodology \citep{DiFonzoGiro2023}, and the different regression-based heuristic approaches proposed by \cite{Yagli2019}, \cite{Kourentzes2019}, and \cite{DiFonzoGiro2022, DiFonzoGiro2023, DiFonzoGiro2023SE}. For more details, the reader is referred to the original articles. The extension to the probabilistic forecast reconciliation setting, which is beyond the scope of this paper, may be pursued according to \cite{Girolimetto2024-jm}.

Let $\Xvet$ be the $[n \times \left(k^\ast + m\right)]$ matrix containing the values to be forecasted of a linearly constrained $n$-variate time series for all the temporal aggregation orders $k \in \mathcal{K}$ (e.g., for a quarterly time series, $\mathcal{K}=\left\{4, 2, 1\right\}$, where $k=4$ corresponds to the annual frequency, $k=2$ semi-annual, and $k=1$ quarterly, wich is the highest observed frequency) where $m = \displaystyle\max_k \mathcal{K}$ and $k^\ast = \displaystyle \sum_{k \in \mathcal{K} \backslash {k_1}} \frac{m}{k} $. Matrix $\Xvet$ can be written as:
$$
\Xvet =
\begin{bmatrix}
	\Xvet^{[m]} & \Xvet^{[k_{p-1}]}&  \ldots & \Xvet^{[1]}
\end{bmatrix}
=
\begin{bmatrix}
	\Uvet^{[m]} & \Uvet^{[k_{p-1}]}&  \ldots & \Uvet^{[1]}\\
	\Bvet^{[m]} & \Bvet^{[k_{p-1}]}&  \ldots & \Bvet^{[1]}
\end{bmatrix}    =
\begin{bmatrix} \Uvet \\ \Bvet \end{bmatrix},
$$
where $\Xvet^{[k]} = \left[\begin{array}{c}\Uvet^{[k]} \\ \Bvet^{[k]}\end{array}\right]$, $k \in \mathcal{K}$, is the $\left(n \times \displaystyle\frac{m}{k}\right)$ matrix containing the target values of the $n$ temporally aggregated series of order $k$. The variables are split into a $[n_u \times \left(k^\ast + m\right)]$ matrix $\Uvet$ and a $[n_b \times \left(k^\ast + m\right)]$ matrix $\Bvet$. These matrices represent, respectively, the upper ($n_u$) and bottom ($n_b$) level variables at any temporal aggregation order within the cross-sectional hierarchy, where the total number of variables is given by $n = n_u+n_b$. It is worth noting that in the cross-temporal framework, the `very bottom' variables are the highest-frequency bottom time series in matrix $\Bvet^{[1]}$.

\begin{table}[!t]
	\centering
	\begin{tabular}{@{\hskip 0.15cm}l@{\hskip 0.3cm}|@{\hskip 0.15cm}l@{\hskip 0.3cm}|@{\hskip 0.15cm}c@{\hskip 0.3cm}|cc}
		\toprule
		\multicolumn{3}{c|}{Regression-based reconciliation: structural representation} & \multicolumn{2}{c}{Constraints}\\
		\textbf{Matrix form}  & \textbf{Vector form} & \textbf{Reconciled} & \quad \textbf{cs}\quad\null  & \null\quad\textbf{te}\quad\null \\
		\midrule
		\addlinespace[0.2cm]
		$\widehat{\Xvet} = \Svet_{\text{cs}}\Bvet + \Evet$ & $\xvet = \Kvet_{\text{cs}}\bvet + \etavet$ &
		$\widetilde{\xvet}_{\text{cs}} = \Kvet_{\text{cs}}\Gvet_{\text{cs}}\widehat{\xvet}$ & \textbf{yes} & {\color{red}\textbf{no}} 
		\\[0.1cm]
		$\widehat{\Xvet} = \Xvet^{[1]}\Svet_{\text{te}}^\top + \Evet$ & $\xvet = \Kvet_{\text{te}}\xvet^{[1]} + \etavet$  &
		$\widetilde{\xvet}_{\text{te}} = \Kvet_{\text{te}}\Gvet_{\text{te}}\widehat{\xvet}$ & {\color{red}\textbf{no}} & \textbf{yes} \\[0.1cm]
		$\widehat{\Xvet} = \Svet_{\text{cs}}\Bvet^{[1]}\Svet_{\text{te}}^\top + \Evet$ & $\xvet =\Kvet_{\text{ct}}\bvet^{[1]} + \etavet$ &
		$\widetilde{\xvet}_{\text{ct}} = \Kvet_{\text{ct}}\Gvet_{\text{ct}}\widehat{\xvet}$ & \textbf{yes} & \textbf{yes} \\
		\addlinespace[0.1cm]
		\bottomrule
	\end{tabular}
	\caption{Matrix and vectorized forms of the structural representation of an $n$-variate linearly constrained time series considered at all temporal aggregation orders, in the cross-sectional (cs), temporal (te), and cross-temporal (ct) frameworks, respectively. The column `Reconciled' reports the formula of the reconciled forecasts coherent with only cross-sectional ($\widetilde{\xvet}_{\text{cs}}$), only temporal ($\widetilde{\xvet}_{\text{te}}$), and both ($\widetilde{\xvet}_{\text{ct}}$) constraints, respectively.}
	\label{tab:strucformulae}
\end{table}

Now, we denote $\Svet_{\text{cs}}$ and $\Svet_{\text{te}}$ the structural (summing) matrices in the cross-sectional and temporal frameworks, respectively,  mapping the relevant bottom variables into their complete counterpart vector. The structural representations valid for all $n$ variables and all temporal aggregation orders $k \in \mathcal{K}$, which take into account either cross-sectional or temporal coherence, are given by $\Xvet = \Svet_{\text{cs}}\Bvet$ and $\Xvet = \Xvet^{[1]}\Svet_{\text{te}}^\top$, respectively. Since $\Xvet^{[1]} = \Svet_{\text{cs}}\Bvet^{[1]}$, and $\Bvet = \Bvet^{[1]}\Svet_{\text{te}}^\top$, we can express in matrix form the cross-temporal structural representation as well: $\Xvet = \Svet_{\text{cs}}\Bvet^{[1]}\Svet_{\text{te}}^\top$. Moreover, it is convenient to consider the vectorized forms of the structural representations above:
$$
\xvet = \Kvet_{\text{cs}}\bvet, \quad \xvet = \Kvet_{\text{te}}\xvet^{[1]}, \quad \xvet = \Kvet_{\text{ct}}\bvet^{[1]},
$$
where $\xvet = \text{vec}(\Xvet^\top)$, $\bvet = \text{vec}(\Bvet^\top)$, $\xvet^{[1]} = \text{vec}({\Xvet^{[1]}}^\top)$,
$\bvet^{[1]} = \text{vec}({\Bvet^{[1]}}^\top)$, and 
$$
\Kvet_{\text{cs}} = \left(\Svet_{\text{cs}} \otimes \Ivet_{(k^\ast + m)}\right), \quad
\Kvet_{\text{te}} = \left(\Ivet_{n} \otimes \Svet_{\text{te}}\right), \quad
\Kvet_{\text{ct}} = \left(\Svet_{\text{cs}} \otimes \Svet_{\text{te}}\right) = \Kvet_{\text{cs}}\Kvet_{\text{te}},
$$
with $\Kvet_{\text{ct}}$ the cross-temporal structural matrix, encompassing both contemporaneous and temporal constraints. The results shown so far are summarized in \autoref{tab:strucformulae} along with the expression for the reconciled forecasts through the classical formula (\citealp{Hyndman2011}):
\begin{equation}
	\label{eq:recformula}
	\widetilde{\xvet}_{\text{r}} = \Kvet_{\text{r}}\Gvet_{\text{r}}\widehat{\xvet}, \quad \text{r} \in \{\text{cs}, \text{te}, \text{ct}\},
\end{equation}
with $\Gvet_{\text{r}} = \left(\Kvet_{\text{r}}^\top\Sigmavet^{-1}\Kvet_{\text{r}}\right)^{-1}\Kvet_{\text{r}}^\top\Sigmavet^{-1}$, where $\Sigmavet$ is a $\left[n\left(k^\ast + m\right) \times n\left(k^\ast + m\right)\right]$ base forecasts' error covariance matrix for the entire system of involved variables, at any cross-sectional and temporal aggregation levels, and $\Mvet_{\text{r}} = \Kvet_{\text{r}}\Gvet_{\text{r}}$ are oblique projection matrices  \citep{Panagiotelis2021}. More precisely, $\widetilde{\xvet}_{\text{cs}}$ is the oblique projection in the space $\mathcal{S}_{\text{cs}}$ spanned by the columns of $\Kvet_{\text{cs}}$, and $\widetilde{\xvet}_{\text{te}}$ is the oblique projection in the space $\mathcal{S}_{\text{te}}$ spanned by the columns of $\Kvet_{\text{te}}$. In addition, $\widetilde{\xvet}_{\text{ct}}$ is the oblique projection in the space $\mathcal{S}_{\text{ct}}$, given by the intersection of $\mathcal{S}_{\text{cs}}$ and $\mathcal{S}_{\text{te}}$ (i.e., $\mathcal{S}_{\text{ct}} = \mathcal{S}_{\text{cs}} \cap \mathcal{S}_{\text{te}}$), spanned by the columns of $\Kvet_{\text{ct}}$. These three approaches are optimal in the least squares sense, but only $\widetilde{\xvet}_{\text{ct}}$ is cross-temporally coherent.

In addition to the structural representation, forecast reconciliation can be formulated as a linear constrained quadratic program \citep{Stone1942, Byron1978, Byron1979, Solomou1993, DiFonzoGiro2023}, as described in \autoref{tab:projformulae}, starting from the zero-constrained representation
$$
\widehat{\Xvet} = \Xvet + \Evet \quad \text{s.t.} \; \Cvet_{\text{cs}}\Xvet = \Zerovet_{n_u \times (k^\ast+m)} \; \text{and/or} \; \Xvet\Cvet_{\text{te}}^\top = \Zerovet_{n \times k^\ast},
$$
where $\Cvet_{\text{cs}}$ and $\Cvet_{\text{te}}$ denote the zero-constraints matrices representing the relationships between the variables in cross-sectional and temporal frameworks, respectively. Also in this case, we may consider the vectorized forms of the constraints:
$$
\Hvet_{\text{cs}}\xvet = \Zerovet_{n_u (k^\ast+m) \times 1}, \quad \Hvet_{\text{te}}\xvet = \Zerovet_{n k^\ast \times 1}, \quad \Hvet_{\text{ct}}\xvet = \Zerovet_{(n_um+ nk^\ast) \times 1} ,
$$
where
$$
\Hvet_{\text{cs}} = \left(\Cvet_{\text{cs}} \otimes \Ivet_{(k^\ast + m)}\right), \quad
\Hvet_{\text{te}} = \left(\Ivet_{n} \otimes \Cvet_{\text{te}}\right), 
$$
and
$$
\Hvet_{\text{ct}} = \begin{bmatrix}
	\left[\Zerovet_{n_um\times nk^\ast} \quad \Ivet_m \otimes \Cvet_{\text{cs}}\right]\\
	\Ivet_{n} \otimes \Cvet_{\text{te}}
\end{bmatrix}
$$
is the cross-temporal matrix encompassing both contemporaneous and temporal constraints.

\begin{table}[t]
	\centering
	\begin{tabular}{l|l|l|cc}
		\toprule
		\multicolumn{3}{c|}{Regression-based reconciliation: zero-constrained representation} & \multicolumn{2}{c}{Constraints}\\
		\textbf{Matrix form}  & \textbf{Vector form} & \textbf{Reconciled} &  \textbf{cs} & \textbf{te} \\
		\midrule
		\addlinespace[0.2cm]
		\multicolumn{1}{l|}{$\widehat{\Xvet} = \Xvet + \Evet$} & \multicolumn{1}{l|}{$\widehat{\xvet} = \xvet + \eta$} & \multicolumn{1}{l|}{} & \multicolumn{1}{l}{} \\[0.25cm]
		$\; s.t. \; \Cvet_{\text{cs}}\Xvet = \Zerovet_{n_u \times (k^\ast+m)}$ & 
		$\; s.t. \; \Hvet_{\text{cs}}\xvet = \Zerovet_{n_u (k^\ast+m) \times 1}$ &
		$\widetilde{\xvet}_{\text{cs}} = \Mvet_{\text{cs}}\widehat{\xvet}$ & \textbf{yes} & {\color{red}\textbf{no}} 
		\\[0.1cm]
		$\; s.t. \; \Xvet\Cvet_{\text{te}}^\top = \Zerovet_{n \times k^\ast}$ & 
		$\; s.t. \; \Hvet_{\text{te}}\xvet = \Zerovet_{n k^\ast \times 1}$ &
		$\widetilde{\xvet}_{\text{te}} = \Mvet_{\text{te}}\widehat{\xvet}$ & {\color{red}\textbf{no}} & \textbf{yes} \\[0.1cm]
		$\; s.t. \; \begin{cases}
			\Cvet_{\text{cs}}\Xvet = \Zerovet_{n_u \times (k^\ast+m)}\\
			\Xvet\Cvet_{\text{te}}^\top = \Zerovet_{n \times k^\ast}
		\end{cases}$ & 
		$\; s.t. \; \Hvet_{\text{ct}}\xvet = \Zerovet_{(n_um+ nk^\ast) \times 1}$ &
		$\widetilde{\xvet}_{\text{ct}} = \Mvet_{\text{ct}}\widehat{\xvet}$ & \textbf{yes} & \textbf{yes} \\
		\addlinespace[0.1cm]
		\bottomrule
	\end{tabular}
	\caption{Matrix and vectorized forms of the zero-constrained  representation of an $n$-variate linearly constrained time series considered at all temporal aggregation orders, in the cross-sectional (cs), temporal (te), and cross-temporal (ct) frameworks. The column `Reconciled' reports the formula of the reconciled forecasts coherent with only cross-sectional ($\widetilde{\xvet}_{\text{cs}}$), only temporal ($\widetilde{\xvet}_{\text{te}}$), and both ($\widetilde{\xvet}_{\text{ct}}$) constraints, respectively.}
	\label{tab:projformulae}
\end{table}

We can express the reconciled forecasts through the formula proposed by \cite{Stone1942} \citep{Byron1978, Byron1979, DiFonzoGiro2023},
\begin{equation}
	\label{eq:constrained_rec_formula}
	\widetilde{\xvet}_{\text{r}} = \Mvet_{\text{r}}\widehat{\xvet}, \quad \text{r} \in \{\text{cs}, \text{te}, \text{ct}\},
\end{equation}
where $\Mvet_{\text{r}} = \left[\Ivet_{n(k^\ast+m)} - \Sigmavet \Hvet_{\text{r}}^\top\left(\Hvet_{\text{r}}\Sigmavet_{\text{r}}\Hvet_{\text{r}}^\top\right)^{-1}\Hvet_{\text{r}}\right]$ is a projection matrix. The formulae to compute $\widetilde{\xvet}_{\text{cs}}$ and $\widetilde{\xvet}_{\text{te}}$ have been used by \cite{Yagli2019} in two heuristic sequential reconciliation approaches, respectively called \textit{cross-sectional-then-temporal} (cst) and \textit{temporal-then-cross-sectional} (cts, see \citealp{DiFonzoGiro2023}), given by:
$$
\widetilde{\xvet}_{\text{cst}} = \Mvet_{\text{te}}\widetilde{\xvet}_{\text{cs}} = \Mvet_{\text{te}}\Mvet_{\text{cs}}\widehat{\xvet}, \quad
\widetilde{\xvet}_{\text{tcs}} = \Mvet_{\text{cs}}\widetilde{\xvet}_{\text{te}} = \Mvet_{\text{cs}}\Mvet_{\text{te}}\widehat{\xvet},
$$
where the forecasts in $\widetilde{\xvet}_{\text{cst}}$ ($\widetilde{\xvet}_{\text{tcs}}$) are not cross-sectionally (temporally) coherent.

\autoref{thm:ite} in \autoref{sec:theorems} establishes sufficient conditions under which the equality $\widetilde{\xvet}_{\text{cst}} = \widetilde{\xvet}_{\text{tcs}} = \widetilde{\xvet}_{\text{seq}}$ holds. Furthermore, it is shown that under the same conditions, $\widetilde{\xvet}_{\text{seq}}$ is equivalent to an optimal combination cross-temporal forecast reconciliation solution with a specific error covariance matrix.

A second heuristic proposed by \cite{Kourentzes2019} is able to obtain cross-temporally coherent forecasts starting from either of the two sequential approaches so far. In particular, the technique consists in an ensemble forecasting procedure that exploits the simple averaging of different forecasts:
\begin{align*}
	\widetilde{\xvet}_{ka_{\text{cst}}} &= (\overline{\Mvet}_{\text{te}}\otimes\Ivet_{n})\widetilde{\xvet}_{\text{cs}} = \overline{\Mvet}_{\text{te}}\Mvet_{\text{cs}}\widehat{\xvet}, \\
	\widetilde{\xvet}_{ka_{\text{tcs}}} &= (\Ivet_{k^\ast+m}\otimes\overline{\Mvet}_{\text{cs}})\widetilde{\xvet}_{\text{te}} = \overline{\Mvet}_{\text{cs}}\Mvet_{\text{te}}\widehat{\xvet},
\end{align*}
with 
$$
\overline{\Mvet}_{\text{cs}} = \displaystyle\frac{1}{p}\sum_{k \in \mathcal{K}} \Mvet^{[k]}  \quad \text{and} \quad  \overline{\Mvet}_{\text{te}} = \displaystyle\frac{1}{n}\sum_{i=1}^n \Mvet_i,
$$ where $\Mvet^{[k]}$ is the cross-sectional projection matrix for each temporal aggregation order and $\Mvet_i$ is the temporal projection matrix for each single variable \cite[see][]{DiFonzoGiro2023}. 

Finally,  \cite{DiFonzoGiro2023} proposed an iterative cross-temporal reconciliation approach alternating forecast reconciliation along one single dimension (cross-sectional or temporal) in a cyclic fashion, until a convergence criterion is met:
$$
\widetilde{\xvet}_{ite_{\text{cst}}} = \lim_{j \rightarrow \infty} \left(\Mvet_{\text{te}}\Mvet_{\text{cs}}\right)^j\widehat{\xvet}, \quad
\widetilde{\xvet}_{ite_{\text{tcs}}} = \lim_{j \rightarrow \infty} \left(\Mvet_{\text{cs}}\Mvet_{\text{te}}\right)^j\widehat{\xvet}.
$$
It is worth noting that, provided convergence was achieved in both cases, $\widetilde{\xvet}_{ite_{\text{cst}}}$ and $\widetilde{\xvet}_{ite_{\text{tcs}}}$ are cross-temporally coherent, but in general $\widetilde{\xvet}_{ite_{\text{cst}}} \ne \widetilde{\xvet}_{ite_{\text{tcs}}}$.

In \autoref{sec:theorems}, \autoref{thm:ite} establishes sufficient conditions under which the equality $\widetilde{\xvet}_{ka_{\text{cst}}} = \widetilde{\xvet}_{ka_{\text{tcs}}} = \widetilde{\xvet}_{\text{seq}}$ holds, and \autoref{thm:ite2} shows a sufficient conditions under which the equality $\widetilde{\xvet}_{ite_{\text{cst}}} = \widetilde{\xvet}_{ite_{\text{tcs}}} = \widetilde{\xvet}_{ite}$ holds, and $\widetilde{\xvet}_{ite}$ converges in norm to the optimal combination forecast reconciliation solution with a specific error covariance matrix.

\section{Two new results}
\label{sec:theorems}
The first result (\autoref{thm:ite}) shows that if the error covariance matrices used in either steps of an iterative reconciliation approach are constant across levels and time granularities, (i) fully reconciled forecasts are obtained in a single (two-step) iteration, (ii) the final result does not depend on the order of application of the uni-dimensional reconciliation approaches, and (iii) the solution is equivalent to that obtained through an optimal combination approach using a separable covariance matrix with a Kronecker product structure. 

\vskip0.5cm
\begin{theorem}\label{thm:ite}
	Let $\Wvet^{[k]}_j$, $k \in \mathcal{K}$, $j = 1, \ldots, \displaystyle\frac{m}{k}$, be the $(n \times n)$ cross-sectional hierarchy error covariance matrix, and $\Omegavet_i$, $i = 1,\ldots,n$, the $[(k^\ast + m) \times (k^\ast + m)]$ error covariance matrix for the $i$-th series temporal hierarchy. If $\Wvet^{[k]}_j = \Wvet$, $\forall k \in \mathcal{K}$ and $\forall j = 1, \ldots, \displaystyle\frac{m}{k}$, and $\Omegavet_i = \Omegavet$, $\forall i = 1,\ldots,n$, then:
	\begin{enumerate}
		\item the iterative cross-temporal forecast reconciliation approach reduces to a single (two-step) iteration. Furthermore, $\widetilde{\Xvet}_{\text{tcs}} = \widetilde{\Xvet}_{\text{cst}} = \widetilde{\Xvet}_{\text{seq}}$, where $\widetilde{\Xvet}_{\text{tcs}}$ ($\widetilde{\Xvet}_{\text{cst}}$) is the [$n\times (k^\ast+m)$] matrix of the \textit{temporal-then-cross-sectional} (\textit{cross-sectional-then-temporal}) reconciled forecasts;
		\item the heuristic cross-temporal reconciliation approach by \cite{Kourentzes2019} is equivalent to a sequential approach;
		\item denoting $\widetilde{\Xvet}_{\text{oct}}$  the optimal combination reconciled forecasts obtained by using the covariance matrix  $\Sigmavet_{\text{ct}} = \Wvet \otimes \Omegavet$, it is $ \widetilde{\Xvet}_{\text{seq}} = \widetilde{\Xvet}_{\text{oct}}$.
	\end{enumerate}
\end{theorem}

\begin{proof} \;\
	\begin{enumerate}
		\item If $\Wvet^{[k]}_j = \Wvet$, $\forall k \in \mathcal{K}$ and $\forall j = 1, \ldots, \displaystyle\frac{m}{k}$, and $\Omegavet_i = \Omegavet$, $\forall i = 1,\ldots,n$,
		the cross-sectional and temporal reconciled forecasts are respectively given by
		$\widetilde{\Xvet}_{\text{cs}}=\Mvet_{\text{cs}}^\ast\widehat{\Xvet}$, and
		$\widetilde{\Xvet}_{\text{te}}=\widehat{\Xvet}{\Mvet_{\text{te}}^\ast}^\top$, where $\Mvet_{\text{cs}}^\ast$ and $\Mvet_{\text{te}}^\ast$ are the $(n \times n)$ and $\left[\left(k^\ast + m\right) \times \left(k^\ast + m\right)\right]$, respectively, oblique projection matrices
		$$
		\Mvet_{\text{cs}}^\ast = 
		\Svet_{\text{cs}}\left(\Svet_{\text{cs}}^\top\Wvet^{-1}\Svet_{\text{cs}}\right)^{-1}\Svet_{\text{cs}}^\top\Wvet^{-1}, \quad
		\Mvet_{\text{te}}^\ast = \Svet_{\text{te}} \left(\Svet_{\text{te}}^\top\Omegavet^{-1}\Svet_{\text{te}}\right)^{-1}\Svet_{\text{te}}^\top\Omegavet^{-1} .
		$$
		It easy to check that
		$\widetilde{\Xvet}_{\text{cst}} = \Mvet_{\text{cs}}^\ast \widehat{\Xvet} {\Mvet_{\text{te}}^\ast}^\top = \widetilde{\Xvet}_{\text{tcs}}$.
		In addition, as both $\widetilde{\Xvet}_{\text{cst}}$ and $\widetilde{\Xvet}_{\text{tcs}}$ are cross-temporally coherent, it is
		$\widetilde{\Xvet}_{\text{tcs}} = \widetilde{\Xvet}_{\text{cst}} = \widetilde{\Xvet}_{\text{seq}}$.
		\item Under the conditions of the theorem we have: $\widetilde{\Xvet}_{ka_{\text{tcs}}}=\overline{\Mvet}_{\text{cs}}\widehat{\Xvet}{\Mvet_{\text{te}}^\ast}^\top$ and $\widetilde{\Xvet}_{ka_{\text{cst}}}=\Mvet_{\text{cs}}^\ast\widehat{\Xvet}\;\overline{\Mvet}_{\text{te}}^\top$.
		It follows
		that
		$$
		\overline{\Mvet}_{\text{cs}} = \displaystyle\frac{1}{p}\sum_{k \in \mathcal{K}} \Mvet^{[k]} = \Mvet_{\text{cs}}^\ast, \quad
		\overline{\Mvet}_{\text{te}} = \displaystyle\frac{1}{n}\sum_{i=1}^n \Mvet_i= \Mvet_{\text{te}}^\ast ,
		$$
		and then 
		$$
		\begin{cases}	
			\widetilde{\Xvet}_{ka_{\text{tcs}}} = \widetilde{\Xvet}_{\text{tcs}}\\
			\widetilde{\Xvet}_{ka_{\text{cst}}} = \widetilde{\Xvet}_{\text{cst}}	
		\end{cases}
		\quad 
		{\Longrightarrow} \quad \widetilde{\Xvet}_{ka_{\text{tcs}}} = \widetilde{\Xvet}_{ka_{\text{cst}}} = \widetilde{\Xvet}_{\text{seq}}.
		$$
		
		\item 
		Recalling that $\Svet_{\text{ct}} = \Svet_{\text{cs}} \otimes \Svet_{\text{te}}$, after a few algebraic steps we find that the optimal combination cross-temporal reconciled forecasts may be expressed as:
		$$
		\widetilde{\xvet}_{\text{oct}} = 
		\Svet_{\text{ct}} \left(\Svet_{\text{ct}}^\top \Sigmavet_{\text{ct}}^{-1}\Svet_{\text{ct}}\right)^{-1}\Svet_{\text{ct}}^\top\Sigmavet_{\text{ct}}^{-1}\widehat{\xvet} =
		\left(\Mvet_{\text{cs}}^\ast \otimes \Mvet_{\text{te}}^\ast\right)\widehat{\xvet} .
		$$
		On the other hand, 
		as
		$\widetilde{\xvet}_{\text{cs}} = \left(\Mvet_{\text{cs}}^\ast \otimes \Ivet_{k^\ast+m}\right)\widehat{\xvet}$
		and
		$\widetilde{\xvet}_{\text{te}} = \left(\Ivet_n \otimes \Mvet_{\text{te}}^\ast\right)\widehat{\xvet}$, 
		it follows
		$$
		\widetilde{\xvet}_{\text{tcs}} = \widetilde{\xvet}_{\text{cst}} = \left(\Mvet_{\text{cs}}^\ast \otimes \Mvet_{\text{te}}^\ast\right)\widehat{\xvet} =
		\widetilde{\xvet}_{\text{oct}}.
		$$
	\end{enumerate}
\end{proof}

In general, when an iterative reconciliation procedure `converges' to an optimal combination reconciliation approach, this happens regardless the order with which each uni-dimensional reconciliation is applied. A well-known example is the iterative approach where the diagonal covariance matrices computed using the one-step-ahead in-sample forecast errors are used in the cross-sectional phase, and the diagonal matrix ‘‘which contains estimates of the in-sample one-step-ahead error variances across each level’’ (\citealp{Athanasopoulos2017}, p. 64) in the temporal one. \autoref{fig:norm} shows the Frobenius norm of the difference between the reconciled forecasts obtained through the optimal $oct$ (with cross-temporal series variance scaling covariance) and iterative heuristic $ite_{\text{tcs}}$ (alternating the temporal series variance scaling and the cross-sectional series variance reconciliation approaches) formulae. The convergence is shown based on 350 replications of the forecasting experiment detailed in \autoref{sec:forexp}, with different convergence tolerance levels ($\delta = 10^{-5}$ in red, $\delta = 10^{-6}$ in green, and $\delta = 10^{-10}$ in blue). We obtain similar results using $ite_{\text{cst}}$.

These empirical observation is theoretically justified by \autoref{thm:ite2}, that presents a sufficient condition under which an iterative reconciliation heuristic converges to an optimal combination approach.

\begin{figure}[tb]
	\centering
	\includegraphics[width=\linewidth]{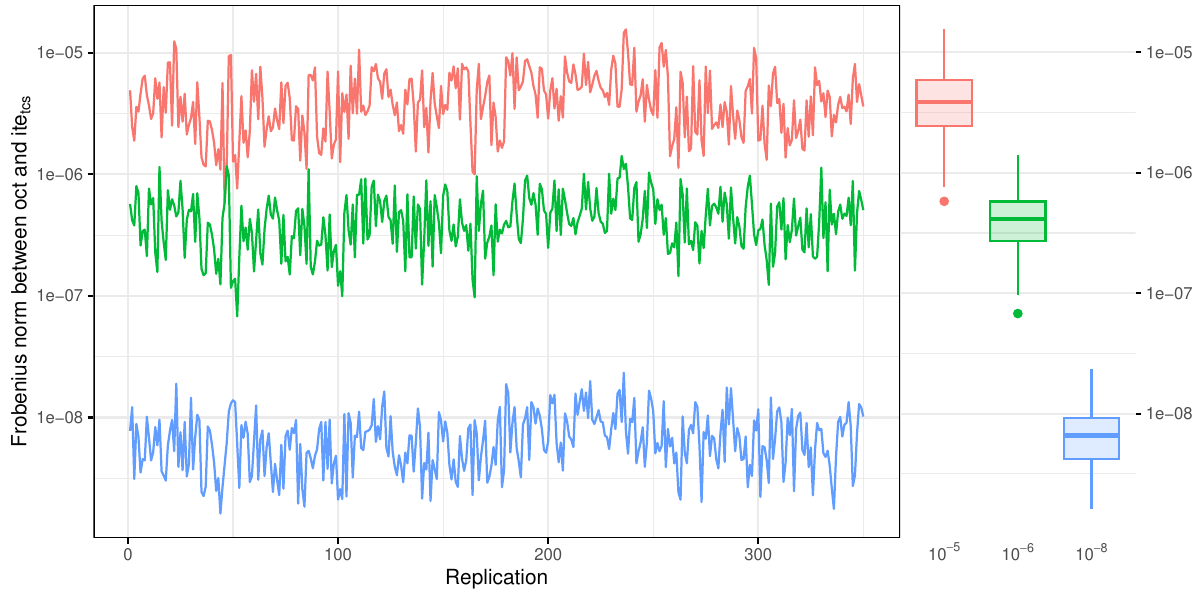}
	\caption{Frobenius norm of the difference between the matrices of the reconciled forecasts using iterative heuristic $ite_{\text{tcs}}$ alternating the temporal series variance scaling and the cross-sectional series variance approaches, converges to $oct$ with the cross-temporal series variance scaling covariance in the empirical application (350 replications of the forecasting experiment) described in \autoref{sec:forexp}. The different convergence tolerance ($\delta$) values are reported with the red ($\delta = 10^{-5}$), green ($\delta = 10^{-6}$) and blue ($\delta = 10^{-10}$).}
	\label{fig:norm}
\end{figure}

\vskip0.25cm
\begin{theorem}\label{thm:ite2}
	Let $\Wvet^{[k]}$, $k \in \mathcal{K}$, be the cross-sectional hierarchy error covariance matrix with $\Sigmavet_{\text{cs}} = \Pvet^\top\mbox{diag}(\Wvet^{[m]}, \dots, \Wvet^{[1]})\Pvet$, and $\Omegavet_i$, $i = 1,\ldots,n$, the $i$-th series temporal hierarchy error covariance matrix with $\Sigmavet_{\text{te}} = \mbox{diag}(\Omegavet_1, \dots, \Omegavet_n)$, with $\Pvet$ the commutation matrix operating on matrices of dimension $[n \times (k^\ast + m)]$ (e.g., $\Xvet$), such that $\Pvet\mbox{vec}\left(\Xvet\right)=\mbox{vec}\left(\Xvet^\top\right)$. If $\Sigmavet_{\text{cs}} = \Sigmavet_{\text{te}}$, then the iterative reconciled forecasts converge in norm to the optimal combination reconciled forecasts with cross-temporal covariance matrix $\Sigmavet_{\text{ct}}$. 
\end{theorem}

\begin{proof}
	Let $\widetilde{\Xvet}_{\text{tcs}}$ and $\widetilde{\Xvet}_{\text{cst}}$ be the $[n \times (k^\ast + m)]$ matrix of the \textit{temporal-then-cross-sectional} and \textit{cross-sectional-then-temporal} iterative reconciled forecasts using $\Wvet^{[k]}$ ($k \in \mathcal{K}$) and $\Omegavet_i$ ($i = 1, \dots, n$) as the cross-sectional and temporal covariance matrices, respectively. The iterative solution obtained at a given complete recursion $J \ge 1$, may be written as the result of alternating oblique projections, such that $\widetilde{\xvet}_{\text{tcs}} = (\Mvet_{\text{cs}}\Mvet_{\text{te}})^J\widehat{\xvet}$, with
	\begin{equation}
		\label{eq:MM}
		\begin{aligned}
			\Mvet_{\text{cs}} &= \Kvet_{\text{cs}} \left[\Kvet_{\text{cs}}^\top \Sigmavet_{\text{cs}}^{-1} \Kvet_{\text{cs}}\right]^{-1} \Kvet_{\text{cs}}^\top\Sigmavet_{\text{cs}}^{-1}\\
			\Mvet_{\text{te}} &= \Kvet_{\text{te}} \left[\Kvet_{\text{te}}^\top \Sigmavet_{\text{te}}^{-1} \Kvet_{\text{te}}\right]^{-1} \Kvet_{\text{te}}^\top\Sigmavet_{\text{te}}^{-1}
		\end{aligned} .
	\end{equation}
	
	\noindent We want to prove that
	$
	\lVert \widetilde{\xvet}_{\text{tcs}} - \widetilde{\xvet}_{\text{oct}}\rVert = \lVert (\Mvet_{\text{cs}}\Mvet_{\text{te}})^J\widehat{\xvet} - \Mvet_{\text{ct}}\widehat{\xvet}\rVert \xrightarrow{J \rightarrow +\infty} 0 
	$.
	
	\noindent Since $\Sigmavet_{\text{te}} = \Sigmavet_{\text{cs}} = \Sigmavet_{\text{ct}} = \Sigmavet$, where $\Sigmavet$ is p.d., we can express its inverse as $\Sigmavet^{-1} = \Qvet^\top\Qvet$, where $\Qvet$ is a (unique) upper triangular matrix with positive diagonal elements. Thus, denoting $\overline{\xvet} = \Qvet\widehat{\xvet}$, it can be shown that
	$$
	\widetilde{\xvet}_{\text{cs}} = \Qvet^{-1} \overline{\Mvet}_{\text{cs}}\overline{\xvet} , \quad
	\widetilde{\xvet}_{\text{te}} = \Qvet^{-1} \overline{\Mvet}_{\text{te}}\overline{\xvet} , \quad
	\widetilde{\xvet}_{\text{oct}} = \Qvet^{-1} \overline{\Mvet}_{\text{ct}}\overline{\xvet} ,
	$$
	where $\overline{\Mvet}_{\text{r}} = \overline{\Kvet}_{\text{r}} \big(\overline{\Kvet}_{\text{r}}^\top \overline{\Kvet}_{\text{r}}\big)^{-1} \overline{\Kvet}_{\text{r}}^\top$, with $\overline{\Kvet}_{\text{r}} = \Qvet\Kvet_{\text{r}}$, $\text{r} \in \{\text{cs}, \; \text{te}\}$, and $\overline{\Mvet}_{\text{ct}} = \overline{\Svet}_{\text{ct}}\big(\overline{\Svet}_{\text{ct}}^\top \overline{\Svet}_{\text{ct}}\big)^{-1} \overline{\Svet}_{\text{ct}}^\top$, with $\overline{\Svet}_{\text{ct}} = \Qvet\Svet_{\text{ct}}$. Matrices $\overline{\Mvet}_{\text{r}}$, $\text{r} \in \{\text{cs}, \; \text{te}, \; \text{ct}\}$ are \textit{orthogonal} projection matrices onto the linear sub-spaces spanned by the columns of, respectively, $\Qvet\Svet_{\text{ct}}$, $\Qvet\Kvet_{\text{cs}}$, and $\Qvet\Kvet_{\text{te}}$. The oblique projection matrices (\ref{eq:MM}) can thus be written using the orthogonal projection matrices $\overline{\Mvet}_{\text{r}}$, $\text{r} \in \{\text{cs}, \; \text{te}\}$, such that $\Mvet_{\text{r}} = \Qvet^{-1} \overline{\Mvet}_{\text{r}}$.
	As $(\Mvet_{\text{cs}}\Mvet_{\text{te}})^J\widehat{\xvet} = \Qvet^{-1}(\overline{\Mvet}_{\text{cs}}\overline{\Mvet}_{\text{te}})^J\overline{\xvet}$, it follows
	$$	
	\widetilde{\xvet}_{\text{tcs}} - \widetilde{\xvet}_{\text{oct}} = \Qvet^{-1}\left[(\overline{\Mvet}_{\text{cs}}\overline{\Mvet}_{\text{te}})^J - \overline{\Mvet}_{\text{ct}}\right]\Qvet\widehat{\xvet}.
	$$
	As the alternating orthogonal projections onto two subspaces converges in norm to the projection onto the intersection of the two subspaces\footnote{Thanks to Tommaso Proietti for suggesting considering alternating projections.}, we have $\lVert\widetilde{\xvet}_{\text{tcs}} - \widetilde{\xvet}_{\text{oct}}\rVert\xrightarrow{J \rightarrow +\infty} 0$. The convergence of $\widetilde{\xvet}_{\text{cst}}$ to $\widetilde{\xvet}_{\text{oct}}$ is easily obtained by inverting the order in which matrices $\Mvet_{\text{cs}}$ and $\Mvet_{\text{te}}$ are applied in the proof. 
\end{proof}

\section{Empirical application}\label{sec:forexp}

The dataset PV324 \citep{Yang2017cs, Yang2017te, Yagli2019, DiFonzoGiro2022, DiFonzoGiro2023SE}
refers to simulated data from 318 photovoltaic (PV) plants located in California. The hourly irradiation data from these PV plants are structured into three hierarchical levels to facilitate various aggregation analyses (\autoref{fig:pv324}).

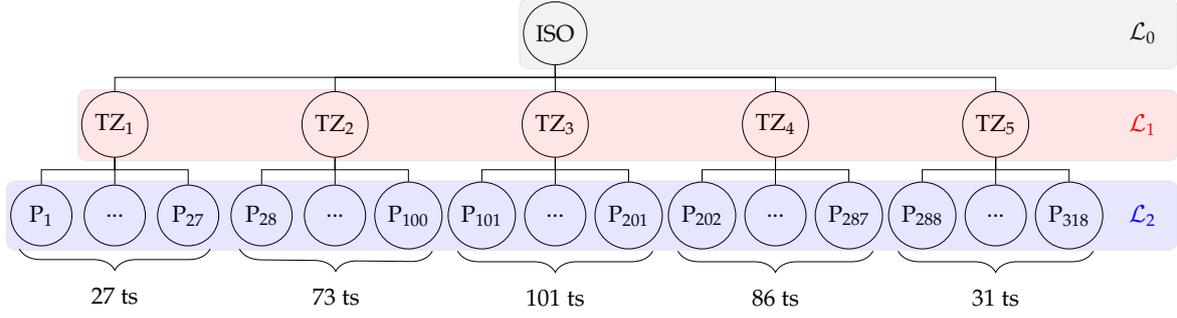
\begin{figure}[!tb]
	\resizebox{\linewidth}{!}{
		\begin{tikzpicture} [baseline=(current  bounding  box.center),
			every node/.append style={shape=circle,
				draw=black,
				minimum size=1cm}
			]
			\node at (0, 0) (A1){P$_{1}$};
			\node at (1.2, 0) (A2){$...$};
			\node at (2.4, 0) (A3){P$_{27}$};
			\node at (1.2, 1.5) (A){TZ$_1$};
			
			\relation{0.2}{A1}{A};
			\relation{0.2}{A2}{A};
			\relation{0.2}{A3}{A};
			
			\node at (3.6, 0) (B1){P$_{28}$};
			\node at (4.8, 0) (B2){$...$};
			\node at (6, 0) (B3){P$_{100}$};
			\node at (4.8, 1.5) (B){TZ$_2$};
			
			\relation{0.2}{B1}{B};
			\relation{0.2}{B2}{B};
			\relation{0.2}{B3}{B};
			
			\node at (7.2, 0) (C1){P$_{101}$};
			\node at (8.4, 0) (C2){$...$};
			\node at (9.6, 0) (C3){P$_{201}$};
			\node at (8.4, 1.5) (C){TZ$_3$};
			
			\relation{0.2}{C1}{C};
			\relation{0.2}{C2}{C};
			\relation{0.2}{C3}{C};
			
			\node at (10.8, 0) (D1){P$_{202}$};
			\node at (12, 0) (D2){$...$};
			\node at (13.2, 0) (D3){P$_{287}$};
			\node at (12, 1.5) (D){TZ$_4$};
			
			\relation{0.2}{D1}{D};
			\relation{0.2}{D2}{D};
			\relation{0.2}{D3}{D};
			
			\node at (14.4, 0) (E1){P$_{288}$};
			\node at (15.6, 0) (E2){$...$};
			\node at (16.8, 0) (E3){P$_{318}$};
			\node at (15.6, 1.5) (E){TZ$_5$};
			
			\relation{0.2}{E1}{E};
			\relation{0.2}{E2}{E};
			\relation{0.2}{E3}{E};
			
			\node at (8.4, 3) (T){ISO};
			\relation{0.2}{A}{T};
			\relation{0.2}{B}{T};
			\relation{0.2}{C}{T};
			\relation{0.2}{D}{T};
			\relation{0.2}{E}{T};
			
			\node[draw = none] at (18, 3) (L0){$\mathcal{L}_0$};
			\node[draw = none, color =red] at (18, 1.5) (L1){$\mathcal{L}_1$};
			\node[draw = none, color = blue] at (18, 0) (L2){$\mathcal{L}_2$};
			
			\node[shape=rectangle, opacity=0.1, rounded corners, inner sep=6pt, fill = blue, fit=(A1.south west)(L2.north east)](Bcs){};
			
			\node[shape=rectangle, opacity=0.1, rounded corners, inner sep=6pt, fill = red, fit=(A.south west)(L1.north east)](Bcs){};
			\node[shape=rectangle, opacity=0.1, rounded corners, inner sep=6pt, fill = gray, fit=(T.south west)(L0.north east)](Bcs){};
			
			\draw [decorate,decoration={brace,amplitude=10pt,mirror, raise = 0.5em}] (A1.south west) -- (A3.south east) node[draw=none, midway,yshift=-2.5em]{27 ts};
			\draw [decorate,decoration={brace,amplitude=10pt,mirror, raise = 0.5em}] (B1.south west) -- (B3.south east) node[draw=none, midway,yshift=-2.5em]{73 ts};
			\draw [decorate,decoration={brace,amplitude=10pt,mirror, raise = 0.5em}] (C1.south west) -- (C3.south east) node[draw=none, midway,yshift=-2.5em]{101 ts};
			\draw [decorate,decoration={brace,amplitude=10pt,mirror, raise = 0.5em}] (D1.south west) -- (D3.south east) node[draw=none, midway,yshift=-2.5em]{86 ts};
			\draw [decorate,decoration={brace,amplitude=10pt,mirror, raise = 0.5em}] (E1.south west) -- (E3.south east) node[draw=none, midway,yshift=-2.5em]{31 ts};
	\end{tikzpicture}}
	\caption{PV324 hierarchy: Independent System Operator level (1 time series, $\mathcal{L}_0$), Transmission Zones level (5 time series, $\mathcal{L}_1$), and Plant level (318 time series, $\mathcal{L}_2$).\label{fig:pv324}}
\end{figure}

At the highest level, $\mathcal{L}_0$, there is a single time series representing the total irradiation data for the Independent System Operator (ISO), which is derived by summing the irradiation data from all 318 PV plants. The intermediate level, $\mathcal{L}_1$, consists of five time series, each representing one of the Transmission Zones (TZ) within California. These TZ time series are aggregated from groups of 27, 73, 101, 86, and 31 PV plants, respectively. At the most granular level, $\mathcal{L}_2$, there are 318 time series, each corresponding to the hourly irradiation data of an individual PV plant.

To assess the new results presented so far, we consider the forecasting experiment pursued by \cite{DiFonzoGiro2023SE} aligned with the operational forecasting requirements of the public corporation managing power grid operations in California \citep[CAISO,][]{Makarov2011, Kleissl2013}, ensuring that the experiment closely mimics real-world conditions for grid management and power forecasting operations\footnote{A complete set of results is available at the GitHub repository \githuburl}. The forecasts were generated with a rolling window of 14 days (336 hours) as training set for a horizon of two days (48 hours), with a specific focus on evaluating the performance of the day 2 (operating day). The base forecasts for the 318 plant-level ($\mathcal{L}_2$) time series were derived from numerical weather prediction (NWP) forecasts provided by 3TIER \citep{3TIER2010}. For the six aggregated time series at $\mathcal{L}_0$ and $\mathcal{L}_1$ levels and for the $\mathcal{L}_2$ time series at various temporal aggregations (2, 3, 4, 6, 8, 12, and 24 hours), the exponential smoothing state space model (ETS) from the R-package \texttt{forecast} \citep{forecast2021} was employed to generate the base forecasts. We consider a total of 3 benchmarks:
\begin{itemize}[leftmargin = !, labelwidth=\widthof{3TIER$_{\text{bu}}$}, itemsep = 1.5mm, topsep = 1.5mm]
	\item[$base$] NWP forecasts for the 318 plant-level time series and ETS for all the aggregated series;
	\item[PERS$_{\text{bu}}$] the forecast at any given hour is equal to the observed value at same hour on the previous day and the forecasts for any temporal granularity are obtained through cross-temporal bottom-up;
	\item[3TIER$_{\text{bu}}$] the NWP forecasts provided by 3TIER are used for the hourly aggregated level, cross-temporal bottom-up is used for all the other levels.
\end{itemize}

Standard forecast reconciliation techniques can sometimes result in negative values. This is problematic when dealing with variables that are naturally non-negative, like PV power generation, as negative values are illogical in such scenarios. The set-negative-to-zero (sntz) approach proposed by \cite{DiFonzoGiro2023SE} is used to avoid this issue. 
In \autoref{fig:alltm}, we compare computational times and CPU memory allocation across iterative ($ite_{\text{tcs}}$ and $ite_{\text{cst}}$), optimal ($oct$), or heuristic (ka$_{\text{tcs}}$ and ka$_{\text{cst}}$) approaches, in the following cases:
\begin{itemize}[leftmargin = !, labelwidth=\widthof{$\;ka_{\text{tcs}}$}, itemsep = 1.5mm, topsep = 1.5mm]
	\item[$ols$] $ite_{\text{tcs/cst}}$ and $ka_{\text{tcs/cst}}$ with $\Wvet = \Ivet_n$ and $\Omegavet = \Ivet_{(k^\ast + m)}$ equivalent to $oct$ with $\Sigmavet_{\text{ct}} = \Ivet_{n(k^\ast+m)}$ (\autoref{thm:ite});
	\item[$str$] $ite_{\text{tcs/cst}}$ and $ka_{\text{tcs/cst}}$ with $\Wvet = \text{diag}(\Svet_{\text{cs}}\Unovet_{n_b})$ and $\Omegavet = \text{diag}(\Svet_{\text{te}}\Unovet_{m})$ equivalent to $\Sigmavet_{\text{ct}} = \Ivet_{n(k^\ast+m)}$ and $oct(str)$ with $\Sigmavet_{\text{ct}} = \text{diag}(\Svet_{\text{ct}}\Unovet_{mn_b})$ (\autoref{thm:ite});
	\item[$str_{\text{cs}}$] $ite_{\text{tcs/cst}}$ and $ka_{\text{tcs/cst}}$ with $\Wvet = \text{diag}(\Svet_{\text{cs}}\Unovet_{n_b})$ and $\Omegavet = \Ivet_{(k^\ast + m)}$ equivalent to $oct$ with $\Sigmavet_{\text{ct}} = \text{diag}(\Svet_{\text{cs}}\Unovet_{n_b}) \otimes \Ivet_{(k^\ast + m)}$ (\autoref{thm:ite});
	\item[$str_{\text{te}}$] $ite_{\text{tcs/cst}}$ and $ka_{\text{tcs/cst}}$ with $\Wvet = \Ivet_n$ and $\Omegavet = \text{diag}(\Svet_{\text{te}}\Unovet_{m})$ equivalent to $oct$ with $\Sigmavet_{\text{ct}} = \Ivet_n \otimes \text{diag}(\Svet_{\text{te}}\Unovet_{m})$ (\autoref{thm:ite});
	\item[$wlsv$] $ite_{\text{tcs/cst}}$ and $ka_{\text{tcs/cst}}$ alternating temporal series variance scaling and cross-sectional series variance reconciliation approaches, and $oct$ using the cross-temporal series variance scaling covariance. Following \autoref{thm:ite2}, $ite_{\text{tcs/cst}}$ converges to $oct$.
\end{itemize}

\noindent  Specifically, Figures \ref{fig:time} and \ref{fig:mem} show the computational performance of the forecasting experiment over 350 replications. Figure \ref{fig:time} illustrates the computational times, measured in seconds, required for each approach, while Figure \ref{fig:mem} shows the corresponding CPU memory usage in megabytes. In both figures, each panel represents a different reconciliation approaches, with the techniques displayed along the y-axes for comparison.

\begin{figure}[!tbh]
	\centering
	\begin{subfigure}[b]{1\linewidth}
		\centering
		\includegraphics[width=\linewidth]{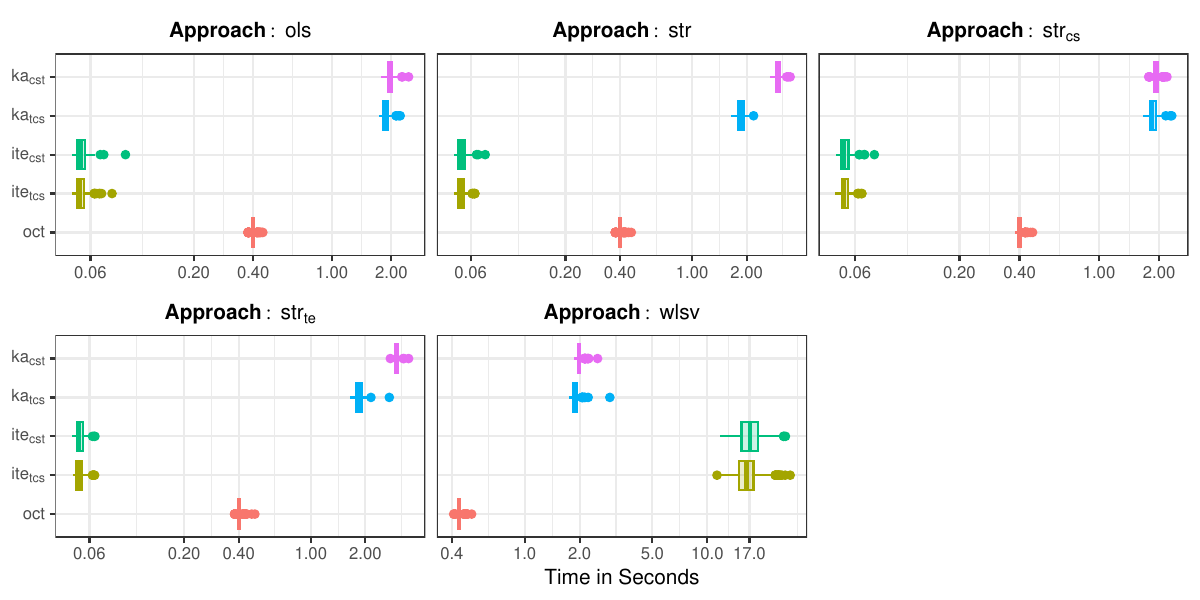}
		\caption{Computational times in seconds.\label{fig:time}}
	\end{subfigure}
	\begin{subfigure}[b]{1\linewidth}
		\centering
		\includegraphics[width=\linewidth]{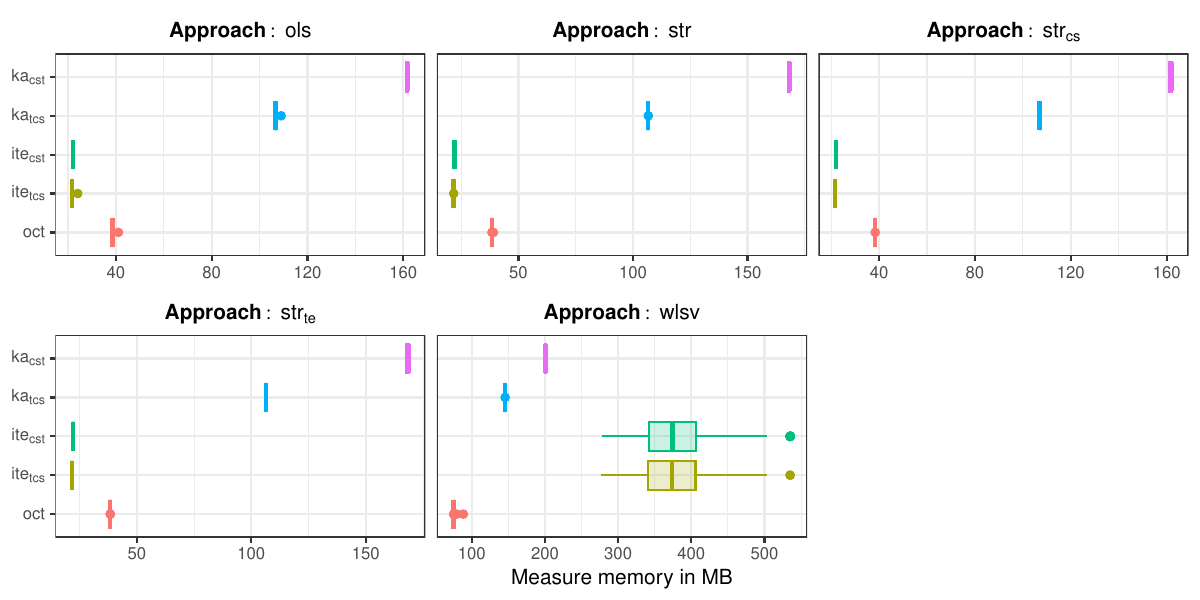}
		\caption{Computational CPU memory in megabyte (MB).\label{fig:mem}}
	\end{subfigure}
	\caption{Computational times of iterative, optimal and heuristic KA reconciliation approaches across different methods ($ols$, $str$, and $wlsv$) for the 350 replications of the forecasting experiment. System’s hardware and software specifications: \texttt{Intel(R) Core(TM) i7-10700 CPU 2.90GHz 2.90 GHz}, \texttt{64 GB RAM}, \texttt{R v4.4.0} and \texttt{FoReco v1.0} \citep{FoReco2024}.}\label{fig:alltm}
\end{figure}

In agreement with \autoref{thm:ite},  $ols$, $str$, $str_{\text{cs}}$, and $str_{\text{te}}$ give identical reconciled forecasts. When using $wlsv$, both iterative heuristics converge to the optimal solution, as stated by \autoref{thm:ite2}, while the KA heuristics ($ka_{\text{tcs}}$ and $ka_{\text{cst}}$) do not. In details, the iterative approaches with a single iteration ($ols$, $str$, $str_{\text{cs}}$, and $str_{\text{te}}$) are faster and use less memory than all other methods. However, when the number of iterations exceeds 1 ($wlsv$), it becomes the most expensive in terms of both CPU memory and time. Overall, the empirical evidence highlights a clear advantage of the optimal approach, emphasizing its computational efficiency and practical suitability compared to the equivalent KA heuristics. The same result holds for the optimal approach $wlsv$ when compared to the corresponding iterative heuristic.

\begin{table}[!tb]
	\centering
	\caption{Forecast accuracy of selected cross-temporal reconciliation approaches and base forecasts in terms of nRMSE(\%) for each temporal aggregation order ($k = 1$ hourly forecasts up to $k = 24$ daily forecasts). Forecast horizon: operating day. Approaches performing worse than PERS$_{\text{bu}}$ are highlighted in red; bold entries and italic entries identify the best and the second best performing approaches, respectively. }\label{tab:nRMSE}
	\begingroup
	\spacingset{1}
	
\begin{tabular}[t]{r|cccccccc}
\toprule
\multicolumn{1}{c}{\textbf{}} & \multicolumn{8}{c}{\textbf{Temporal aggregation orders}} \\
\textbf{App.} & $k=1$ & $k=2$ & $k=3$ & $k=4$ & $k=6$ & $k=8$ & $k=12$ & $k=24$\\
\midrule
\addlinespace[0.3em]
\multicolumn{1}{c}{} & \multicolumn{8}{c}{$\mathcal{L}_0$}\\
PERS$_{bu}$ & \textcolor{black}{34.62} & \textcolor{black}{34.14} & \textcolor{black}{33.35} & \textcolor{black}{33.11} & \textcolor{black}{29.45} & \textcolor{black}{32.23} & \textcolor{black}{20.83} & \textcolor{black}{20.23}\\
3TIER$_{bu}$ & \textcolor{black}{26.03} & \textcolor{black}{25.34} & \textcolor{black}{24.56} & \textcolor{black}{23.31} & \textcolor{black}{22.09} & \textcolor{black}{\textbf{19.89}} & \textcolor{black}{15.70} & \textcolor{black}{\textbf{12.48}}\\
$base$ & \textcolor{black}{27.85} & \textcolor{black}{27.55} & \textcolor{black}{28.31} & \textcolor{black}{29.69} & \textcolor{black}{28.01} & \textcolor{red}{35.32} & \textcolor{red}{21.04} & \textcolor{black}{18.17}\\
$ols$ & \textcolor{black}{30.15} & \textcolor{black}{29.40} & \textcolor{black}{28.05} & \textcolor{black}{28.41} & \textcolor{black}{24.87} & \textcolor{black}{27.66} & \textcolor{black}{17.71} & \textcolor{black}{17.22}\\
$str$ & \textcolor{black}{26.64} & \textcolor{black}{26.27} & \textcolor{black}{25.56} & \textcolor{black}{25.48} & \textcolor{black}{22.86} & \textcolor{black}{24.71} & \textcolor{black}{16.25} & \textcolor{black}{15.73}\\
$str_{te}$ & \textcolor{black}{28.18} & \textcolor{black}{27.76} & \textcolor{black}{27.01} & \textcolor{black}{26.91} & \textcolor{black}{24.05} & \textcolor{black}{26.14} & \textcolor{black}{17.06} & \textcolor{black}{16.56}\\
$str_{cs}$ & \textcolor{black}{28.53} & \textcolor{black}{28.02} & \textcolor{black}{26.98} & \textcolor{black}{27.18} & \textcolor{black}{24.07} & \textcolor{black}{26.44} & \textcolor{black}{17.14} & \textcolor{black}{16.64}\\
$wlsv$ & \textcolor{black}{\textbf{23.63}} & \textcolor{black}{\em{23.18}} & \textcolor{black}{\em{22.55}} & \textcolor{black}{\em{22.32}} & \textcolor{black}{\em{20.16}} & \textcolor{black}{21.40} & \textcolor{black}{\em{14.30}} & \textcolor{black}{13.64}\\
ka$_{tcs}$ & \textcolor{black}{\em{23.63}} & \textcolor{black}{\textbf{23.06}} & \textcolor{black}{\textbf{22.43}} & \textcolor{black}{\textbf{22.19}} & \textcolor{black}{\textbf{20.04}} & \textcolor{black}{\em{21.25}} & \textcolor{black}{\textbf{14.23}} & \textcolor{black}{\em{13.55}}\\
ka$_{cst}$ & \textcolor{black}{24.21} & \textcolor{black}{23.77} & \textcolor{black}{23.14} & \textcolor{black}{22.94} & \textcolor{black}{20.70} & \textcolor{black}{22.05} & \textcolor{black}{14.66} & \textcolor{black}{14.03}\\
\addlinespace[0.3em]
\multicolumn{1}{c}{} & \multicolumn{8}{c}{$\mathcal{L}_1$}\\
PERS$_{bu}$ & \textcolor{black}{43.15} & \textcolor{black}{42.42} & \textcolor{black}{41.27} & \textcolor{black}{40.81} & \textcolor{black}{36.29} & \textcolor{black}{39.19} & \textcolor{black}{25.66} & \textcolor{black}{24.57}\\
3TIER$_{bu}$ & \textcolor{black}{33.95} & \textcolor{black}{32.94} & \textcolor{black}{31.90} & \textcolor{black}{30.62} & \textcolor{black}{28.02} & \textcolor{black}{26.99} & \textcolor{black}{19.87} & \textcolor{black}{\textbf{16.75}}\\
$base$ & \textcolor{black}{34.24} & \textcolor{black}{34.04} & \textcolor{black}{33.79} & \textcolor{black}{35.01} & \textcolor{black}{33.11} & \textcolor{red}{40.51} & \textcolor{black}{25.16} & \textcolor{black}{20.94}\\
$ols$ & \textcolor{black}{38.19} & \textcolor{black}{36.98} & \textcolor{black}{34.83} & \textcolor{black}{35.38} & \textcolor{black}{30.72} & \textcolor{black}{34.05} & \textcolor{black}{21.80} & \textcolor{black}{20.94}\\
$str$ & \textcolor{black}{32.99} & \textcolor{black}{32.41} & \textcolor{black}{31.45} & \textcolor{black}{31.23} & \textcolor{black}{27.96} & \textcolor{black}{29.91} & \textcolor{black}{19.88} & \textcolor{black}{19.00}\\
$str_{te}$ & \textcolor{black}{35.02} & \textcolor{black}{34.38} & \textcolor{black}{33.31} & \textcolor{black}{33.09} & \textcolor{black}{29.52} & \textcolor{black}{31.76} & \textcolor{black}{20.94} & \textcolor{black}{20.05}\\
$str_{cs}$ & \textcolor{black}{35.17} & \textcolor{black}{34.43} & \textcolor{black}{33.03} & \textcolor{black}{33.18} & \textcolor{black}{29.33} & \textcolor{black}{31.89} & \textcolor{black}{20.87} & \textcolor{black}{20.01}\\
$wlsv$ & \textcolor{black}{\textbf{30.21}} & \textcolor{black}{\em{29.52}} & \textcolor{black}{\em{28.66}} & \textcolor{black}{\em{28.27}} & \textcolor{black}{\em{25.40}} & \textcolor{black}{\em{26.79}} & \textcolor{black}{\em{18.02}} & \textcolor{black}{17.01}\\
ka$_{tcs}$ & \textcolor{black}{\em{30.24}} & \textcolor{black}{\textbf{29.40}} & \textcolor{black}{\textbf{28.53}} & \textcolor{black}{\textbf{28.13}} & \textcolor{black}{\textbf{25.27}} & \textcolor{black}{\textbf{26.64}} & \textcolor{black}{\textbf{17.95}} & \textcolor{black}{\em{16.92}}\\
ka$_{cst}$ & \textcolor{black}{30.82} & \textcolor{black}{30.15} & \textcolor{black}{29.27} & \textcolor{black}{28.91} & \textcolor{black}{25.97} & \textcolor{black}{27.47} & \textcolor{black}{18.41} & \textcolor{black}{17.43}\\
\addlinespace[0.3em]
\multicolumn{1}{c}{} & \multicolumn{8}{c}{$\mathcal{L}_2$}\\
PERS$_{bu}$ & \textcolor{black}{59.75} & \textcolor{black}{56.81} & \textcolor{black}{54.23} & \textcolor{black}{52.87} & \textcolor{black}{46.82} & \textcolor{black}{49.07} & \textcolor{black}{33.12} & \textcolor{black}{30.65}\\
3TIER$_{bu}$ & \textcolor{black}{53.46} & \textcolor{black}{50.57} & \textcolor{black}{48.33} & \textcolor{black}{46.19} & \textcolor{black}{41.36} & \textcolor{black}{40.72} & \textcolor{black}{29.28} & \textcolor{black}{25.19}\\
$base$ & \textcolor{black}{53.46} & \textcolor{black}{47.22} & \textcolor{black}{44.87} & \textcolor{black}{44.30} & \textcolor{black}{39.68} & \textcolor{black}{42.66} & \textcolor{black}{30.92} & \textcolor{black}{25.82}\\
$ols$ & \textcolor{black}{50.69} & \textcolor{black}{47.64} & \textcolor{black}{44.66} & \textcolor{black}{44.53} & \textcolor{black}{38.86} & \textcolor{black}{41.71} & \textcolor{black}{27.64} & \textcolor{black}{25.82}\\
$str$ & \textcolor{black}{46.51} & \textcolor{black}{43.80} & \textcolor{black}{41.80} & \textcolor{black}{40.83} & \textcolor{black}{36.34} & \textcolor{black}{37.90} & \textcolor{black}{25.85} & \textcolor{black}{23.97}\\
$str_{te}$ & \textcolor{black}{47.95} & \textcolor{black}{45.27} & \textcolor{black}{43.22} & \textcolor{black}{42.26} & \textcolor{black}{37.55} & \textcolor{black}{39.38} & \textcolor{black}{26.67} & \textcolor{black}{24.81}\\
$str_{cs}$ & \textcolor{black}{48.84} & \textcolor{black}{46.02} & \textcolor{black}{43.48} & \textcolor{black}{43.08} & \textcolor{black}{37.92} & \textcolor{black}{40.24} & \textcolor{black}{27.00} & \textcolor{black}{25.17}\\
$wlsv$ & \textcolor{black}{\textbf{44.44}} & \textcolor{black}{\em{41.52}} & \textcolor{black}{\em{39.63}} & \textcolor{black}{\em{38.43}} & \textcolor{black}{\em{34.27}} & \textcolor{black}{\em{35.25}} & \textcolor{black}{\em{24.35}} & \textcolor{black}{\em{22.27}}\\
ka$_{tcs}$ & \textcolor{black}{\em{44.48}} & \textcolor{black}{\textbf{41.45}} & \textcolor{black}{\textbf{39.55}} & \textcolor{black}{\textbf{38.34}} & \textcolor{black}{\textbf{34.18}} & \textcolor{black}{\textbf{35.15}} & \textcolor{black}{\textbf{24.30}} & \textcolor{black}{\textbf{22.21}}\\
ka$_{cst}$ & \textcolor{black}{44.76} & \textcolor{black}{41.86} & \textcolor{black}{39.98} & \textcolor{black}{38.80} & \textcolor{black}{34.59} & \textcolor{black}{35.67} & \textcolor{black}{24.56} & \textcolor{black}{22.52}\\
\bottomrule
\end{tabular}

	\endgroup
\end{table}

Following \cite{Yagli2019} and \cite{DiFonzoGiro2023SE}, the accuracy of the considered approaches is measured in terms of normalized Root Mean Square Error (nRMSE):
$$
\mbox{nRMSE}_{i,j}^{[k]} = \frac{\sqrt{\displaystyle\frac{1}{L_k}\sum_{l = 1}^{L_k} \left(\overline{y}_{i,j,l}^{[k]} - y_{i,l}^{[k]}\right)^2}}{\displaystyle\frac{1}{L_k}\sum_{l = 1}^{L_k} y_{i,l}^{[k]}}
$$
where $i = 1,...,n$, denotes the series, $k \in \mathcal{K}$, $\overline{y}_j$ with $j = 0,...,J$, denotes the forecasting approach, and $L_k = 350 \displaystyle\frac{m}{k}$. In \autoref{tab:nRMSE}, we report the nRMSE at different cross-sectional ($\mathcal{L}_0$, $\mathcal{L}_1$ and $\mathcal{L}_2$) and temporal ($k \in \{1,2,3,4,6,8,12,24\}$) levels. Note that we compare the optimal solution for all the approaches ($ols$, $str$, $str_{\text{cs}}$, $str_{\text{te}}$, $wlsv$) as well as the results for ka$_{\text{cst}}$ and ka$_{\text{tcs}}$ alternating the temporal series variance scaling and the cross-sectional series variance matrices. The PERS$_{\text{bu}}$ approach, used as a baseline, almost always has the highest nRMSE values. In contrast, the 3TIER$_{\text{bu}}$ and $str$ approaches demonstrate significant improvements over the base forecasts, especially at higher temporal aggregation orders, as discussed in \cite{DiFonzoGiro2023SE}. $str_{\text{te}}$ and $str_{\text{cs}}$ are always better than $ols$, but worse than $str$. Notably, the $wlsv$ and ka$_{\text{tcs}}$ approaches yield the lowest nRMSE values, reflecting their superior forecasting accuracy across all cross-sectional levels and temporal aggregation orders.

The Multiple Comparison with the Best (MCB) Nemenyi test \citep{koning2005, Kourentzes2019, Makridakis2022} allows to establish if the forecasting performances of the considered techniques are significantly different. This procedure shown in \autoref{fig:nem} evaluates the accuracy of hourly and daily solar irradiation forecasts focusing on the operating day forecast horizon for all the 324 series across $\mathcal{L}_0$, $\mathcal{L}_1$, and $\mathcal{L}_2$ levels. For daily forecasts, the $wlsv$ approach is the most accurate, followed closely by ka$_{\text{tcs}}$ and ka$_{\text{cst}}$. The comparable performance of these three approaches is evident from their overlapping intervals, that means no statistically significant difference in their accuracy. On the other hand, approaches such as PERS$_{\text{bu}}$, $ols$, and $base$ show worse performance compared to $wlsv$, as indicated by the non-overlapping intervals. For hourly forecasts, a similar trend emerges, with ka$_{\text{tcs}}$ showing superior accuracy, immediately followed by $wlsv$ and ka$_{\text{cst}}$. However, the intervals of $wlsv$, ka$_{\text{tcs}}$, and ka$_{\text{cst}}$ slightly overlap, suggesting that their performances are statistically different. It should be noted that even if $str$ does not perform as well as the top three approaches, it consistently outperforms all the remaining approaches, including $base$, PERS$_{\text{bu}}$ and 3TIER$_{\text{bu}}$.

\begin{figure}[!tb]
	\centering
	\includegraphics[width=0.95\linewidth]{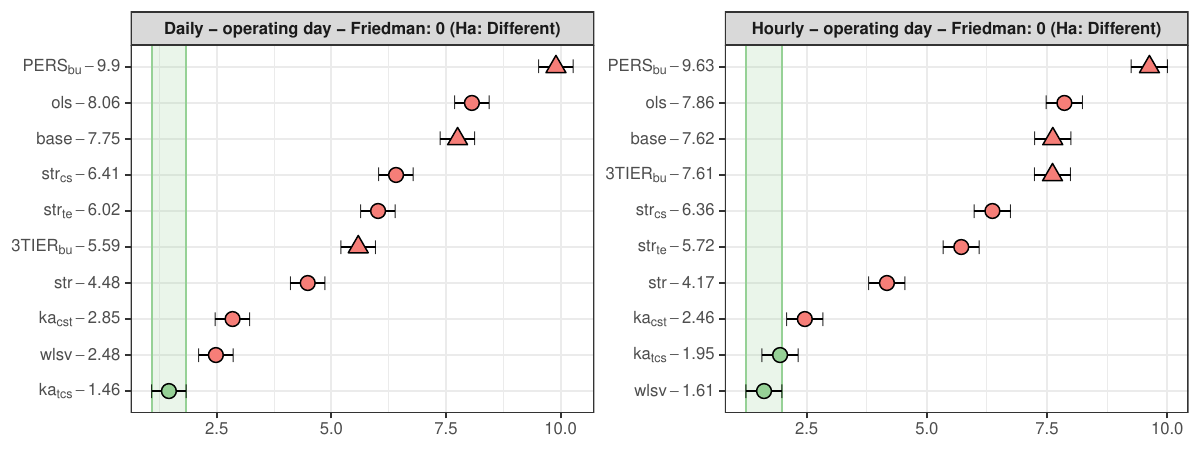}
	\caption{MCB-Nemenyi test on selected cross-temporal reconciliation approaches with operating day forecast horizon. $\mathcal{L}_0$, $\mathcal{L}_1$, $\mathcal{L}_2$ levels (324 series). Right panel: hourly forecasts; Left panel: daily forecasts. The mean rank of each approach is displayed to the right of their names. If the intervals of two forecast reconciliation approaches do not overlap, this indicates a statistically different performance. Thus, approaches that do not overlap with the green interval are considered significantly worse than the best, and vice-versa.}
	\label{fig:nem}
\end{figure}

\section{Conclusion}
\label{sec: conclusion}

This paper investigates the relationships between sequential, iterative and optimal combination cross-temporal forecast reconciliation approaches. Our study establishes conditions under which a sequential approach, whether starting from a cross-sectional or temporal reconciliation, is equivalent to the heuristics proposed by \cite{Kourentzes2019} and \cite{DiFonzoGiro2023}. We found that when error covariance matrices remain consistent across different levels and time granularities, these heuristics give the optimal combination approach solution that employs a separable covariance matrix based on the Kronecker product. Moreover, we demonstrate that specific patterns in the error covariance matrix of the iterative reconciliation converge to the optimal solution, regardless of the order of application of the uni-dimensional reconciliation steps. In addition, we offer a comprehensive framework for understanding and improving cross-temporal forecast reconciliation, taking into consideration not only the forecast accuracy, but also computational aspects that have so far been quite overlooked.

Our empirical application utilizes the PV324 dataset \citep{Yagli2019}, with 324 simulated hourly irradiation data from photovoltaic plants in California, structured into a three-level hierarchy. The experiment employs a rolling window forecasting approach with a 14-day training set and a two-day forecast horizon. The findings indicate that iterative approaches with a single iteration are computationally faster and utilize less memory compared to methods requiring multiple iterations. Moreover, the optimal combination approach demonstrates better trade-off between accuracy and computational efficiency compared to the two heuristics and the different benchmarks.

These results have significant implications for the field of time series forecasting, particularly in applications requiring high-dimensional and hierarchical forecast reconciliation. By highlighting the conditions under which different reconciliation techniques are equivalent, this research provides valuable insights for practitioners aiming to optimize forecasting accuracy and efficiency.

Future research could explore the application of these findings to other types of hierarchical structures and different domains, as well as the potential integration of machine learning techniques \citep{Rombouts2024} to further enhance forecast reconciliation processes.

\section*{Acknowledgments}
The authors thank all the participants to the 52$^{nd}$ Scientific Meeting of the Italian Statistical Society in Bari (Italy) and the 44$^{th}$ International Symposium on Forecasting in Dijon (France). Funding support for this article was provided by the Ministero dell’Università e della Ricerca with the project PRIN2022 “PRICE: A New Paradigm for High Frequency Finance” 2022C799SX.

\phantomsection\addcontentsline{toc}{section}{References}

\bibliographystyle{apalike3.bst}
\bibliography{mybibfile.bib}

\end{document}